\documentclass[leqno]{article}
\usepackage{amsmath}
\usepackage{amsthm}
\usepackage{amssymb}
\usepackage{amscd}
\usepackage{latexsym}
\usepackage{amsfonts}
\usepackage{color}
\input xy
\xyoption{all} \tolerance=500

=1in \textwidth 15.6cm \topmargin 15mm\textheight22.1cm \advance\topmargin by -\headheight\advance\topmargin by
-\headsep

\def\sd{\mathsf{d}}
\def\sV{\mathsf{V}}
\def\sv{\mathsf{v}}
\def\sT{\mathsf{T}}
\def\sP{\mathsf{P}}
\def\R{\mathbb R}
\def\be{\begin{equation}}
\def\ee{\end{equation}}

\newtheorem{prop}{Proposition}
\newtheorem{definition}{Definition}
\newtheorem{corollary}{Corollary}
\newtheorem{example}{Example}

\newdir{|>}{%
!/4.5pt/@{|}*:(1,-.2)@^{>}*:(1,+.2)@_{>}}
\newdir{ (}{{}*!/-5pt/@^{(}}

\begin{document}

\title{Geometry of Routh reduction\thanks{Research founded by the  Polish National Science Centre grant under the contract number DEC-2012/06/A/ST1/00256.}}
\author{Katarzyna Grabowska and Pawe\l\ Urba\'nski\\
Faculty of Physics, University of  Warsaw\\
ul. Pasteura 5, 02-093 Warszawa, Poland}

\maketitle

\begin{abstract}
Routh reduction for Lagrangian systems with cyclic variable is presented as an example of Lagrangian reduction. It appears that Routhian, which is a generating object of reduced dynamics, is not a function any more but a section of a bundle of affine values.
\end{abstract}

\section{Introduction}\label{sec:intro}

The Routh reduction \cite{R} is a classical piece of analytical mechanics. It concerns systems with a special type of symmetry, namely systems  having cyclic variables. Usually Routh reduction is presented in terms of coordinates and a cyclic variable is a coordinate $q$ such that $\frac{\partial L}{\partial q} =0$. There where several attempts to geometrize Routh reductions (see e.g. \cite{M} and \cite{A}), but only partially successful. Most frequent simplifying assumption is that the configuration manifold is a product of two manifolds: with cyclic and non-cyclic variables. In the paper we present a complete geometric framework for Routh reduction using the language of symplectic reductions and their generating objects. We show that Routhian which is the resulting generating object for the reduced dynamics is not a function, but a section of an affine bundle (the bundle of affine values). Such situation we have already encountered in the dynamics of a charged particles, frame independent formulation of Newtonian analytical mechanics and in analytical mechanics of non-autonomous systems (see \cite{TU,U}).

In section \ref{sec:not} we briefly recall geometrical tools we use in the following sections. In particular, in section \ref{sub:affine} we review the rudiments of the {\it av}-geometry, which are essential for a geometric approach to Routh reduction. We follow here \cite{GGU1,GGU2,U}. Since we work in the Tulczyjew approach to analytical mechanics and Legendre transformation \cite{T1,T2,T3,T4} we need the theory of double vector and double affine bundles as well \cite{KU,GR,GRU}.

Section \ref{sec:sym} contains main results of our work. The Routh reduction relation is first obtained as a composition of the Legendre transformation from Lagrangian to Hamiltonian side, a symplectic reduction on the Hamiltonian side, and the affine Legendre transformation from  Hamiltonian  to Lagrangian side. We show then how the reduction can be performed on the Lagrangian side only. Routh reduction is a very instructive example of a Lagrangian reduction, which illustrates the principal differences between Lagrangian and Hamiltonian reductions. The Lagrangian reduction relation, unlike the Hamiltonian one, involves values of generating objects (Lagrangians) of systems.

Section \ref{sec:gen} contains few remarks about possible generalisations of Routh reduction setting.

\section{Notation and preliminaries}\label{sec:not}

Let $Q$ be a smooth differential manifold with coordinates $(q^i)$ in an open subset $\mathcal{O}\subset Q$. In the total space of the {\it tangent bundle} $\tau_Q:\sT Q\rightarrow Q$ we will use adopted coordinates $(q^i, \dot q^k)$ defined in $\tau_Q^{-1}(\mathcal{O})$. Similarly, in the total space of the {\it cotangent bundle} $\pi_Q:\sT^\ast Q\rightarrow Q$ adopted coordinates $(q^i,  p_k)$ are defined in $\pi_Q^{-1}(\mathcal{O})$. Let us recall the structure of tangent and cotangent bundles, since it will be very much used in the following sections.

Tangent and cotangent bundles are pair of dual vector bundles, therefore we expect that any structure compatible with the vector bundle structure has its counterpart on the dual bundle. For example, it is well known that the tangent bundle is a canonical example of a {\it Lie algebroid}. It can be given by the Lie bracket of vector fields together with the identity map as an anchor. The dual counterpart of the Lie algebroid structure on $\sT Q$ is the canonical linear {\it symplectic form} $\omega_Q$ on $\sT^\ast Q$ or, equivalently, the canonical {\it Poisson bivector field} $\Lambda_Q$ on $\sT^\ast Q$. In coordinates, these two objects read
\be\label{not:6}
\omega_Q=\sd p_i\wedge\sd q^i, \qquad \Lambda_Q=\partial_{p_k}\wedge\partial_{q^k}.
\ee
Since we are going to use Tulczyjew approach in mechanics, it will be convenient to encode Lie algebroid structure of $\sT Q$ and symplectic structure on $\sT^\ast Q$ in appropriate morphisms of double vector bundles.

Recall that a {\it double vector bundle} is a manifold equipped with two compatible vector bundle structures  \cite{P,KU}. The compatibility condition can be expressed conveniently as commuting of the two homogeneity structures associated to multiplying vectors by real numbers \cite{GR}. Morphisms of double vector bundles are, naturally, smooth maps linear with respect to both vector bundle structures. Canonical examples of double vector bundles are tangent and cotangent bundles to vector bundles, in particular iterated tangent and cotangent bundles $\sT\sT Q$, $\sT\sT^\ast Q$, $\sT^\ast\sT^\ast Q$ and $\sT^\ast\sT Q$. Since for any vector bundle $E\rightarrow M$ the double vector bundles $\sT^\ast E$ and $\sT^\ast E^\ast$ are canonically isomorphic \cite{D}, so are $\sT^\ast\sT^\ast Q$ and $\sT^\ast\sT Q$. Symplectic structure of $\sT^\ast Q$ can be alternatively given by a map
\be\label{not:7}
\beta_Q:\sT\sT^\ast Q\longrightarrow \sT^\ast\sT^\ast Q, \quad
\beta_Q(v)=\omega_Q(v,\cdot).
\ee
Using the adopted coordinates $(q^i, p_j, \dot q^k, \dot p_l)$ in $\sT\sT^\ast Q$ and $(q^i, p_j, \xi_k, y^l)$ in $\sT^\ast\sT^\ast Q$ we can write
\be\label{not:8}
(q^i, p_j, \xi_k, y^l)\circ\beta_Q=(q^i, p_j, \dot p_k,-\dot q^l).
\ee
Since $\sT^\ast\sT^\ast Q$ is a cotangent bundle itself, it carries a canonical symplectic form
$$\omega_{\sT^\ast Q}=\sd\xi_k\wedge \sd q^k+\sd y^j\wedge \sd p_j.$$
Moreover $\sT\sT^\ast Q$ is also a symplectic manifold with tangent lift of $\omega_Q$, i.e.
$$\sd_{\sT}\omega_Q=\sd p_i\wedge\sd \dot q^i+\sd\dot p_j\wedge\sd q^j.$$
It is easy to check that $\beta_Q$ is a symplectomorphism with respect to these structures.

An algebroid structure on $E\rightarrow M$ can be given as a double vector bundle morphism from $\sT^\ast E$ to $\sT E^\ast$ covering the identity of $E^\ast$ \cite{GU}. For $E=\sT Q$ we usually write it in opposite direction, which is possible, since the appropriate double vector bundle morphism is a diffeomorphism. The Lie algebroid structure of $\sT Q$ is then encoded in the Tulczyjew isomorphism $\alpha_Q$
\be\label{not:9}
\alpha_Q:\sT\sT^\ast Q\longrightarrow \sT^\ast\sT Q.
\ee
Using the adopted coordinates $(q^i, \dot q^j, \varphi_k, \psi_j )$ in $\sT^\ast\sT Q$  we can write
\be\label{not:10}
(q^i, \dot q^j, \varphi_k, \psi_l)\circ\alpha_Q=(q^i, \dot q^j, \dot p_k, p_j).
\ee
In the literature (e.g. in \cite{T3, T4}) $\alpha_Q$ is usually defined as the dual of the canonical flip $\kappa_Q$
\be\label{not:11}
\kappa_Q:\sT\sT Q\longrightarrow \sT\sT Q,\ee
that in coordinates reads
\be
\label{not:12} (q^i, \dot q^j, \delta q^k, \delta\dot q^l)\circ\kappa_Q=(q^i, \delta q^j, \dot q^k,\delta\dot x^l).
\ee
To define $\alpha_Q$ as a dual of $\kappa_Q$ we have to state which vector bundle structure in $\sT\sT Q$ we use. In the source $\sT\sT Q$ we use $\tau_{\sT Q}$ with the obvious dual bundle being $\pi_{\sT Q}:\sT^\ast\sT Q\rightarrow \sT Q$, while in the target $\sT\sT Q$ we use the second vector bundle structure over $\sT Q$, namely $\sT\tau_Q$. The dual bundle can be identified with $\sT\pi_Q:\sT\sT^\ast Q\rightarrow\sT Q$. Parity between elements of $\sT\sT^\ast Q$ and $\sT\sT Q$ with the same tangent projection on $\sT Q$ can be given in terms of curves. For any $w\in \sT\sT^\ast Q$ and $u\in \sT\sT Q$ such that $\sT\pi_Q(w)=\sT\tau_Q(u)$ there exist two curves $t\mapsto p(t)$ and $t\mapsto v(t)$ in $\sT^\ast Q$ and $\sT Q$ respectively, covering the same curve in $Q$ so it makes sense to write $\langle p(t),v(t)\rangle$. The tangent pairing is defined by the following formula
\be\label{not:14}
\langle\!\langle w,u\rangle\!\rangle =\frac{\sd}{\sd t}_{|t=0} \langle p(\cdot),v(\cdot)\rangle.
\ee

The formula connecting $\alpha_Q$ with the bracket of vector fields is complicated. It is more convenient to relate the bracket with $\kappa_Q$. First recall that if $e$, $f$ are elements of the same fiber of $E\rightarrow M$ then we can lift $f$ vertically to $e$. Namely $f^{\sv}_e$ is an element of $\sT_e E$ tangent at $t=0$ to the curve $t\mapsto e+tf$. The formula relating the Lie bracket of vector fields with $\kappa_M$ reads
\be\label{not:13}
\sT X(Y)-\kappa_Q(\sT Y(X))=[X, Y]^\sv_{Y}.
\ee
From the symplectic point of view $\alpha_Q$ is a symplectomorphism with respect to $\sd_{\sT}\omega_Q$ and $\omega_{\sT Q}=\sd \varphi_i\wedge \sd q^i+\sd \psi_j\wedge\sd \dot q^j$.

The Routh reduction relation is constructed out of a vector field on Q. For this purpose we need few elementary facts. Let then $X$ be a vector field on $Q$. The {\it complete lift} of $X$ is a vector field $\sd_{\sT} X$ on $\sT Q$ given by the formula
\be\label{not:1} \sd_{\sT}X(v)=\kappa_Q(\sT X(v)). \ee If $\varphi_t$ denotes the flow of $X$ then $\Phi_t:\sT Q\rightarrow \sT Q$,
$\Phi_t=\sT\varphi_t$ is the flow of $\sd_\sT X$. In adopted coordinates on $\sT Q$ the complete lift of a vector field
$X(q)=X^i(q)\frac{\partial}{\partial q^i}$ reads \be\label{not:2} \sd_{\sT} X(q,\dot q)=X^i(q)\frac{\partial}{\partial q^i}+\frac{\partial
X^j}{\partial q^k}\dot q^k\frac{\partial}{\partial\dot q^j}. \ee

Vector fields can also be lifted to the cotangent bundle. Let $\imath_X$ denote a function on $\sT^\ast Q$, linear in fibres, given by $X$,
\be\label{not:3}\imath_X(p)=\langle\, p, X(\pi_Q(p))\,\rangle. \ee The hamiltonian vector field for the function $\imath_X$ is denoted by
$\sd_{\sT^\ast}X$ and called the {\it cotangent lift} of a vector field $X$. One may use the formula \be\label{not:4}
\sd_{\sT^\ast}X(p)=\beta_Q^{-1}(\sd\imath_X(p)). \ee Flow of the cotangent lift of $X$ is given by $\Phi^\ast_t=\sT^\ast\varphi_{-t}$ while
the coordinate expression for the lift in adopted coordinates reads
\be\label{not:5} \sd_{\sT^\ast}X(q,p)=X^i(q)\frac{\partial}{\partial
q^i}-\frac{\partial X^j}{\partial q^k}\dot p_j\frac{\partial}{\partial p_k}. \ee

There exists another useful characterization of the cotangent lift of a vector field. The field $\sd_{\sT^\ast} X$ is uniquely characterized by the property (see \cite{KU})
$$\langle\!\langle\sd_{\sT^\ast} X(p), \sd_{\sT} X(v)\rangle\!\rangle=0.$$

\bigskip

\subsection{Mechanics}\label{sub:mech}

In this paper we shall use Tulczyjew approach to mechanics as introduced in \cite{T1,T2,T3,T4}. In this approach Lagrangian and Hamiltonian are two different generating objects of the same Lagrangian submanifold $\mathcal{D}$ (the dynamics) of $\sT\sT^\ast Q$. The dynamics $\mathcal{D}$, being a subset of the tangent bundle, is (sometimes implicit) differential equation on curves in the phase space $\sT^\ast Q$. Euler Lagrange equations, traditionally associated with Lagrangian mechanics, are just consequences of phase equations for paths in configuration space. Hamiltonian and Lagrangian formulations of mechanics are equivalent on the infinitesimal level, but Hamiltonian formulation, unlike Lagrangian formulation, has an infinitesimal version only.

Usually, the dynamics $\mathcal{D}$ is generated by a Lagrangian function $L:\sT Q\rightarrow \R$, more precisely
$$\mathcal{D}=\alpha_Q^{-1}(\sd L(\sT Q)).$$
In regular cases, the dynamics can  also be generated by a Hamiltonian function $H:\sT^\ast Q\rightarrow \R$ (depending on convention, Hamiltonian generating object may be plus or minus Hamiltonian function).
$$\mathcal{D}=\beta_Q^{-1}(\sd H(\sT^\ast Q)).$$
There exist systems for which more general generating objects are needed both in Lagrangian and Hamiltonian formalisms. Geometrical structures of Tulczyjew mechanics can be summarized in a diagram called {\it Tulczyjew triple}. Left-hand side corresponds to Hamiltonian approach and right-hand side to Lagrangian.
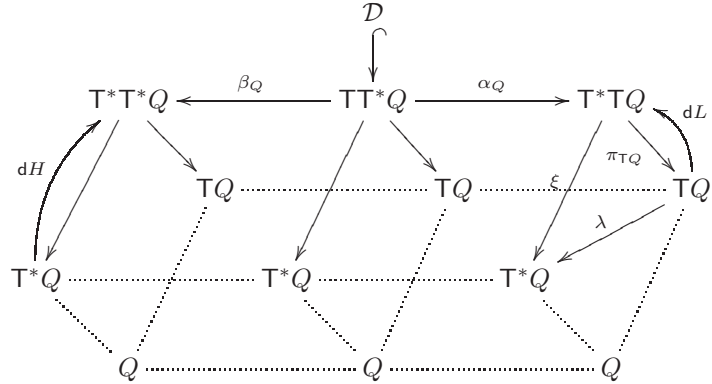
\be\label{not:15} \xymatrix@C-20pt@R-5pt{
 & & & & \mathcal{D}\ar@{ (->}[d] & & & & \\
 & \sT^\ast\sT^\ast Q  \ar[dr] \ar[ddl]
 & & & \sT\sT^\ast Q \ar[rrr]^{\alpha_Q} \ar[dr]\ar[ddl]\ar[lll]_{\beta_Q}
 & & & \sT^\ast\sT Q\ar[dr]_{\pi_{\sT Q}}\ar[ddl]_{\xi} & \\
 & & \sT Q\ar@{.}[rrr]\ar@{.}[ddl] & & & \sT Q\ar@{.}[rrr]\ar@{.}[ddl]
 & & & \sT Q \ar@{.}[ddl]\ar@/_1pc/[ul]_{\sd L}\ar[dll]_{\lambda}\\
 \sT^\ast Q\ar@{.}[rrr]\ar@{.}[dr] \ar@/^1pc/[uur]^{\sd H} & & & \sT^\ast Q\ar@{.}[rrr]\ar@{.}[dr]
 & & & \sT^\ast Q\ar@{.}[dr] & &  \\
 & Q\ar@{.}[rrr]& & & Q\ar@{.}[rrr]& & & Q &
}\ee

The Legendre transformation is understood as the transition from Lagrangian to Hamiltonian formulation of mechanics. On the level of generating objects it involves composition of a Lagrangian with the generating object of the canonical isomorphism $\sT^\ast\sT^\ast Q\simeq \sT^\ast\sT Q$. The general Hamiltonian generating object is the following function
\begin{equation}\label{not:15a}
  \sT Q\times_Q \sT^\ast Q\ni (v,p)\longmapsto F(v,p)=L(v)-\langle\, p,\,v\,\rangle\in\R,
\end{equation}
Which is understood as a family of functions on $\sT^\ast Q$ parameterized by elements of $\sT Q$. In regular cases, i.e. for regular Lagrangians we can get rid of the parameters and simplify $F$ to a Hamiltonian function on $\sT^\ast Q$. The Legendre map, on the on the other hand, associates momenta to velocities and reads
\begin{equation}\label{not:15b}
  \lambda: \sT Q\longmapsto \sT^\ast Q,\qquad \lambda=\xi\circ\sd L.
\end{equation}
For the details of symplectic relations and generating objects one may consult \cite{LM,B}.

There are several generalizations known for Tulczyjew mechanics. One of them is the Tulczyjew version of mechanics on algebroids (see \cite{GGU3,GG}). The generalized version of Tulczyjew triple for mechanics on algebroids is based on the following diagram:
\be\label{not:16} \xymatrix@C-10pt@R-5pt{
 & & & & \mathcal{D}\ar@{ (->}[d] & & & & \\
 & \sT^\ast E^\ast \ar[rrr]^{\tilde\Lambda}  \ar[dr] \ar[ddl]
 & & & \sT E^\ast \ar[dr]\ar[ddl]
 & & & \sT^\ast E\ar[dr]_{\pi_{E}}\ar[ddl]\ar[lll]_{\varepsilon} & \\
 & & E\ar@{.}[rrr]\ar@{.}[ddl] & & & \sT M\ar@{.}[rrr]\ar@{.}[ddl]
 & & & E \ar@{.}[ddl]\ar@/_1pc/[ul]_{\sd L}\ar[dll]_{\lambda}\\
 E^\ast\ar@{.}[rrr]\ar@{.}[dr] \ar@/^1pc/[uur]^{\sd H} & & & E^\ast\ar@{.}[rrr]\ar@{.}[dr]
 & & & E^\ast \ar@{.}[dr] & &  \\
 & M\ar@{.}[rrr]& & & M\ar@{.}[rrr]& & & M &
} \ee
with $\varepsilon$ encoding the structure of an algebroid.

\subsection{Differential geometry of affine values}\label{sub:affine}

In classical mechanics probably the most common tool we use is linear algebra. Vector spaces, vector bundles and linear maps are almost everywhere. It is clear however that there are physical quantities that should not be represented by linear mathematical objects. For example in Newtonian mechanics the transformation rules for momentum associated to the change of inertial frame are of affine character. Therefore, in frame independent formulation momenta should be elements of some kind of affine phase space. In the following we shall recall bits and pieces of the {\it geometry of affine values} that has already been successfully applied to mechanics of charged particles, Newtonian mechanics, time dependent mechanics as well as higher order mechanics and field theory. For the details one may consult \cite{GGU1,GGU2,TU,U,GV}.

Let $\zeta: Z\rightarrow M$ be a $(\R, +)$ principal bundle. It can be regarded also as an affine bundle modelled on the trivial bundle $M\times \R\rightarrow M$. In the literature it is called a {\it bundle of affine values} ({\it av}\,-bundle), since sections of this bundle replace functions in affine differential calculus. We can for example define differentials of sections in the same way as we define differentials of functions. Let us note that the difference of two local sections $\sigma$, $\sigma'$ of $\zeta$ defined on some open subset $\mathcal{O}$ of $M$ is a function on $\mathcal{O}$. Indeed $(\sigma'-\sigma)(x)=\sigma'(x)-\sigma(x)\in\R$. We shall say that $\sigma$ and $\sigma'$ have the same differential at $x\in M$ if $\sd(\sigma-\sigma')(x)=0$. Pairs $(x,\sigma)$ and $(x',\sigma')$ are equivalent if $x=x'$ and $\sigma$ and $\sigma'$ have the same differential at $x$. The set of equivalence classes of pairs $(x,\sigma)$ will be denoted by $\sP Z$ and called {\it affine phase bundle}. The equivalence class of $(x,\sigma)$ will be denoted by $\sd\sigma(x)$. In particular, $\sP(M\times \R) =\sT^\ast M$. It is easy to see that $\sP Z$ is an affine bundle over $M$ modelled on $\sT^\ast M$.  Moreover, it carries the canonical symplectic form $\omega_Z$. To see that, let us observe that we can use the differential $\sd\sigma_0$ of the chosen section $\sigma_0$ to identify $\sP Z$ with $\sT^\ast M$. The form $\omega_Z$ can be given as the pull-back of $\omega_M$ by this identification. In fact $\omega_Z$ does not depend on the choice of $\sigma_0$ since $\omega_M$ is invariant with respect to translation by the differential of a function. The image of a differential $\sd \sigma$ of a
section $\sigma$ is a Lagrangian submanifold of $\sP Z$. A section can be therefore considered {\it an affine generating object}. There are obvious generalizations of affine generating objects like affine Morse families and  sections over submanifolds \cite{U}.

Let $Z_1$ and $Z_2$ be {\it av}\,-bundles over $M_1$ and $M_2$ respectively. An {\it av}\,-relation $R:Z_1\rightarrow Z_2$ is a differential relation such that $z_2\in R(z_1)$ implies $z_2+t\in R(z_1+t)$ and $z_2+s\notin R(z_1)$ for $s\neq 0$. To define the graph of an {\it av}\,-relation we introduce first an {\it av}\,-bundle $Z_2\ominus Z_1$ with the base manifold $M_2\times M_1$. $Z_2\ominus Z_1$ is a manifold of equivalence classes of the following equivalence relation in $Z_2\times Z_1$:
$$(z_2, z_1)\sim (z_2+t, z_1-t), \; t\in\R.$$
It is an immediate observation that $\sP(Z_2\ominus Z_1)=\sP(Z_2)\ominus \sP(Z_1)$, where $\sP(Z_2)\ominus \sP(Z_1)=\sP(Z_2)\times\sP(Z_1)$ as manifolds but with $\omega_{Z_2}-\omega_{Z_1}$ as the associated symplectic form. The graph of $R$ in the category of {\it av}\,-bundles (denoted here by $\mathrm{graph}(R)$) is the image in $Z_2\ominus Z_1$ of the graph of $R$ in the category of sets. $\mathrm{graph}(R)$ is then the image of a section of $Z_2\ominus Z_1$ over the graph of underlying relation $\underbar{R}: M_1\rightarrow M_2$. The phase lift of $R$ is the Lagrangian submanifold of $\sP(Z_2)\ominus\sP(Z_1)$ generated by $\mathrm{graph}(R)$.

The affine phase bundle $\sP Z$ can be obtained also by reduction from $\sT^\ast Z$. Every section $\sigma$ of $\zeta$ corresponds to the function
$$f_\sigma: Z\rightarrow\R,\quad f_\sigma(z)=\sigma(\zeta(z))-z.$$
Differentials of functions of the form $f_\sigma$ fill the coisotropic submanifold $K_1$ of covectors $\varphi$ satisfying $\langle\varphi,\partial_r\rangle=-1$. $\sP Z$ can be identified with reduction of $\sT^\ast Z$ with respect to $K_1$. Characteristics of $K_1$ correspond to orbits of the group action lifted to $\sT^\ast Z$, so that reduction coincides with dividing by the group action. Adopting this point of view we can say that $\sP Z$ is an affine subbundle
in the dual to Atiyah algebroid $A^\ast(Z)$ or that $A^\ast(Z)$ is a vector hull of $\sP Z$. In the context of the geometry of affine values the Atiyah algebroid $A(Z)$ is usually denoted  by $\widetilde\sT Z$. In particular, $\widetilde\sT (M\times\R) = \sT M\times \R$ and the numbers are values of the pairing of vectors and covectors. This observation justifies the following definition of the pairing between $\sP Z$ and $\sT M$.

\begin{definition}\label{def:1}
Let $(p,v)\in \sP Z\times_M \sT M$,  then $\langle p,v\rangle$ is an element of $\widetilde\sT Z$ represented by $\sT\sigma (v)$, where $\sigma$ is a section of $Z$ representing $p$.
\end{definition}

Let us observe that $\widetilde\sT Z$ is itself a bundle of affine values over $\sT M$ with the group action given by distinguished element $\partial_r$ in $\widetilde\sT Z$, i.e. invariant vector field generating group action on $Z$. Adding a number $s$ to an element of $\widetilde\sT Z$ means adding $s$ times the vector $\partial_r$.

Affine phase bundle $\sP Z$ and Atiyah algebroid $\widetilde\sT Z$ replace $\sT^\ast M$ and $\sT M\times \R$ in an affine version of the Tulczyjew triple \cite{U}
\be\label{not:21}
\xymatrix@C-20pt@R-5pt{
% & & & & \mathcal{D}\ar@{ (->}[d] & & & & \\
 & \sT^\ast\sP Z  \ar[dr] \ar[ddl]
 & & & \sT\sP Z \ar[rrr]^{\alpha} \ar[dr]\ar[ddl]\ar[lll]_{\beta}
 & & & \sP \widetilde\sT Z\ar[dr]\ar[ddl] & \\
 & & \sT M\ar@{.}[rrr]\ar@{.}[ddl] & & & \sT M\ar@{.}[rrr]\ar@{.}[ddl]
 & & & \sT M \ar@{.}[ddl]\ar@/_1pc/[ul]_{\sd L}\ar[dll]_{\lambda}\\
 \sP Z\ar@{.}[rrr]\ar@{.}[dr] \ar@/^1pc/[uur]^{\sd H} & & & \sP Z\ar@{.}[rrr]\ar@{.}[dr]
 & & & \sP Z\ar@{.}[dr] & &  \\
 & M\ar@{.}[rrr]& & & M\ar@{.}[rrr]& & & M &
}
\ee
The map $\beta$ comes from the symplectic form $\omega_Z$ while map $\alpha$ is reduced form $\alpha_Z$. The relation between $\sT\sT^\ast Z$ and $\sT \sP Z$ is by symplectic reduction with respect to $\sT K_1\subset\sT \sT^\ast Z$ and symplectic form $\sd_{\sT}\omega_Z$. The relation between $\sT^\ast\sT Z$ and $\sP \widetilde\sT Z$ is by symplectic reduction with respect to $\alpha_Z(\sT K_1)$ and canonical symplectic form on $\sT^\ast\sT Z$.

\section{Reductions and symmetries}\label{sec:sym}

Routh reduction in its classical formulation concerns mechanical system possessing cyclic variables which means that Lagrangian of the system does not depend on one or more configuration coordinates. If we want to work with geometric objects rather than specific coordinates we have to express this feature differently. Our starting point will be a system with configuration
manifold $Q$. The symmetry will be encoded in a smooth non-vanishing vector field on $Q$. For technical reasons we shall assume that the flow of $X$ is regular in a sense that the set  $Q_X$  of trajectories of $X$ is a manifold. The map $Q\rightarrow Q_X$, the trajectory through a point $q$ to the point $q$ will be denoted by $\zeta$. The fact that the mechanical system with Lagrangian $L:\sT Q\rightarrow \R$ has the symmetry $X$ is given by the equation
\be\label{rr:1}
\sd_{\sT} X (L)=0,
\ee
which is the simplest geometrical way of expressing the idea of `a Lagrangian not depending on one of the configuration coordinates'. We do not assume $Q$ to be a Cartesian product as it is usually done for other versions of Routh reduction.

Traditionally Routh reduction provides us with a function called Ruthian that plays the role of new, reduced Lagrangian. In the following we adopt a little different point of view. Starting from $Q$ and $X$ we construct the {\it Routh reduction relation} which gives us the reduced phase space and the procedure to obtain reduced generating object for dynamics. This reduction relation is independent on the particular system, i.e. particular Lagrangian. It can be then applied to any system with symmetry. In fact it can be applied to any system with configuration manifold $Q$, but the results of this reduction are valuable only for systems with symmetry. If a Lagrangian satisfies (\ref{rr:1}) then the behaviour of the full system can be recovered from the behaviour of the reduced one.

The theory that we are going to present is based on symplectic relations. Since any symplectic relation is a Lagrangian submanifold of the product symplectic manifold (in this case mostly cotangent bundles and their affine versions), we will be using the language of generating functions for these submanifolds. Of course, sometimes it is necessary to use more general generating objects as functions on submanifolds or Morse families \cite{B,LM}.

Since Ruthian is a reduced Lagrangian, then the reduction should be in principle performed on the Lagrangian side of the Tulczyjew triple. However, it is the theory of Hamiltonian reduction which is well established and understood. We therefore first pass to Hamiltonian side of the Tulczyjew triple and perform the reduction there and then come back to the Lagrangian mechanics again. The direct Lagrangian reduction will be discussed in section \ref{sub:rutlag}. In section \ref{sub:alg} we comment on reduction to mechanics on Lie algebroid.

\subsection{Routh reduction relation via Hamiltonian mechanics.}\label{sub:rutham}

Lagrangian systems and Hamiltonian systems with configuration manifold $Q$ are described by dynamics $\mathcal{D}$ which is a first order differential equation on curves in the phase space $\sT^\ast Q$. In the following we will find the reduced phase space and the describe the nature of the reduced generating object for the reduced dynamics.
Passing from Lagrangian to Hamiltonian mechanics means using Legendre transformation. On the level of Lagrangian submanifolds it means that we apply the symplectomorphism $\gamma_Q:\sT^\ast\sT Q\longrightarrow\sT^\ast\sT^\ast Q$ to $\sd L(\sT Q) \subset \sT^\ast\sT Q$. On the level of generating objects we have to add a generating object of the symplectomorphism $\gamma_Q$ to the generating object of the dynamics and (possibly) reduce it. The symplectomorphism $\gamma_Q$ is generated by the evaluation function defined on the submanifold of vectors and covectors over the same point
$$\sT^\ast Q\times\sT Q\supset \sT^\ast Q\times_Q\sT Q\ni(p,v)\longmapsto -\langle p,\,v\rangle\in\R.$$
Adding this to the Lagrangian we get the  Hamiltonian Morse family
\begin{equation}\label{rut:1}\xymatrix{
\sT^\ast Q\times_Q\sT Q \ar[d]\ar[r] & \R \\
\sT^\ast Q & },\quad (p,v)\longmapsto L(v)-\langle p,\,v\rangle.\end{equation}
In some cases the family can be reduced to a simpler one, or even to one function called Hamiltonian. Since we are going to define Routh reduction relation first we shall skip the Lagrangian and work with relations between symplectic manifolds.

In Hamiltonian description of the dynamics, the symmetry is given by the cotangent lift $\sd_{\sT^\ast}X$ of the vector field $X$. We know then that phase trajectories of any system with this symmetry lie in level sets of the Hamiltonian function for $\sd_{\sT^\ast}X$ i.e. $\imath_X$. The level sets are coisotropic submanifolds of $\sT^\ast Q$ of codimension $1$. Let us then choose a value $\alpha\in\R$ and define $C_\alpha=\imath_X^{-1}(\alpha)$.  Since $\imath_X$ is a linear function, $C_\alpha$ is fiber-by-fiber over $Q$ an affine subspace of the appropriate fibre of the cotangent bundle. The model vector space for the affine space $C_\alpha\cap\sT_q^\ast Q$ is the intersection $C_0\cap\sT_q^\ast Q$. Symplectic reduction with respect to $C_\alpha$ is the same as dividing $C_\alpha$ by a cotangent lift of $(\R,+)$-action given by the flow of $X$. What we get is a symplectic reduction relation
$$\sT^\ast Q\supset C_\alpha\longrightarrow P_\alpha$$
with $P_\alpha$ being the reduced symplectic manifold. Since lifted group action is linear (i.e. the group acts by linear maps between fibres) the reduced symplectic manifold also has affine character. More precisely it is an affine bundle over $Q_X$. Every fibre is modelled on $C_0$ divided by $\R$-action which clearly is isomorphic $\sT^\ast Q_X$. We can see now that $P_\alpha\rightarrow Q_X$ looks like an affine
phase bundle. Let us find the appropriate {\it av}\,-bundle.

In $Q\times\R$ we define an equivalence relation
\begin{equation}\label{rut:1b}
(q,r)\simeq (\varphi_s(q), r+s\alpha)
\end{equation}
The quotient manifold $Z_\alpha$ is a bundle over $Q_X$. Moreover it is an {\it av}\,-bundle with the action $[q,r]+t:=[q, r+t]=[\varphi_s(q), r+t+s\alpha]$. Associated to the relation we have the map
\begin{equation}\label{rut:1a}
\zeta_\alpha: Q\times \R\ni(q,r)\longrightarrow [q,r]\in Z_\alpha
\end{equation}
which is an {\it av}\,- bundle morphism.

\begin{prop}\label{prop:1a} The manifolds $P_\alpha$ and $\sP Z_\alpha$ are isomorphic as affine bundles and symplectomorphic.
\end{prop}

\noindent{\bf Proof:} Every element of $Z_\alpha$ corresponds to a function on appropriate fibre of the projection $Q\rightarrow Q_X$, with the property that differentiated in the direction of $X$ gives $\alpha$. A section $\sigma$ of $Z_\alpha$ then corresponds to a function $f_\sigma$ on $Q$, such that $\imath_X(\sd f_\sigma(q))=\alpha$.
It means that $\sd f_\sigma(q)\in C_\alpha$.  Two sections $\sigma_1$ and $\sigma_2$ have the same differential at $x\in Q_X$ if and only if differentials of functions $f_{\sigma_1}$ and $f_{\sigma_2}$ lie on the same orbit of lifted $\R$-action. It means that every orbit of the lifted action in $C_\alpha$ corresponds to one class of sections. We have then the well defined map $\sP Z_\alpha\rightarrow P_\alpha$. This map covers the identity on $Q_X$ and is affine. One may check by direct calculation that this map is also a symplectomorphism.  $\Box$

The phase space of the reduced system can be then identified with $\sP Z_\alpha$. The phase lift of the reduction $\sT^\ast Q\supset C_\alpha\rightarrow \sP Z_\alpha$ is a symplectic relation from $\sT^\ast\sT^\ast Q$ to $\sT^\ast\sP Z_\alpha$ generated by the function
equal to zero defined on the graph of the base relation, i.e. on $C_\alpha\times_{Q_X}\sP Z_\alpha$. Adding this generating object to the
previous one, i.e. to (\ref{rut:1}), results in cutting the domain. The Hamiltonian generating object for reduced dynamics now reads
\begin{equation}\label{rut:2}
\xymatrix{
C_\alpha\times_Q\sT Q \ar[d]\ar[r] & \R \\
\sP Z_\alpha & },\quad (p,v)\longmapsto L(v)-\langle p,\,v\rangle.\end{equation} The dynamics is a submanifold of $\sT \sP Z_\alpha$. Looking at the diagram (\ref{not:21}) we see that the Lagrangian generating object of this dynamics should be a section or a family of sections of the {\it av}\,-bundle $\widetilde\sT Z_\alpha\rightarrow\sT Q_X$. Passing to that generating object means performing the Legendre transformation once again, this time in the opposite
direction. Generating object of the isomorphism between $\sT^\ast\sP Z_\alpha$ and $\sP\widetilde\sT Z_\alpha$ is a section of the {\it av}\,-bundle $\sP Z_\alpha\times \widetilde\sT Z_\alpha\longrightarrow \sP Z_\alpha\times \sT Q_X$ over the submanifold  $\sP Z_\alpha\times_{Q_X} \sT
Q_X$ given by affine pairing between tangent vectors and affine covectors (see Definition \ref{def:1})
$$\sP Z_\alpha\times \sT Q_X\supset \sP Z_\alpha\times_{Q_X} \sT Q_X\ni (p_\alpha, w)\longmapsto \langle p_\alpha, w\rangle\in \widetilde\sT Z_\alpha.$$
Final Lagrangian generating object of the reduced dynamics is the following family of sections
\begin{equation}\label{rut:3}
\xymatrix{
\sT Q_X\times_{Q_X} C_\alpha\times_Q\sT Q\ar[d]\ar[r] & \widetilde\sT Z_\alpha \\
\sT Q_X & },\quad
(w,p,v)\longmapsto L(v)-\langle p,v\rangle+\langle p_\alpha, w\rangle.
\end{equation}
This family is the composition of the Lagrangian and the family
\begin{equation}\label{rut:4}
\xymatrix{
\sT Q_X\times_{Q_X} C_\alpha\times_Q\sT Q\ar[d]\ar[r] & \widetilde\sT Z_\alpha \\
\sT Q_X  \times_{Q_X}\sT Q & },\quad
(w,p,v)\longmapsto -\langle p,v\rangle+\langle p_\alpha, w\rangle.
\end{equation}
generating the relation $\sT^\ast\sT Q \rightarrow \sP \widetilde\sT Z_\alpha$. We simplify (\ref{rut:4}) looking  for stationary points in the direction of $(C_\alpha)_q$, i.e. we keep $q, w, v$ fixed and differentiate in the vertical directions in $C_\alpha$.
%%%%%%%%%%%%%%%%%%%%%%%%%%%%%%%%%%%
This means we change $p$ by an element $\varphi\in (C_0)_q$ and, consequently, $p_\alpha$ by an appropriate element $\varphi_\alpha$ of $\sT^\ast_x Q_X$, with $q$ over $x$. The condition for stationary point is then
$$\langle \varphi_\alpha, w\rangle=\langle \varphi,v\rangle$$
which is true for any $p\in C_\alpha$ if $v$ projects on $w$ while dividing $\sT Q$ by lifted $\R$-action.
The generating family (\ref{rut:4}) simplifies to the section of {\it av}\,-bundle $\sT Q\times_{\sT Q_X}\widetilde\sT Z_\alpha\rightarrow \sT Q$ given by the formula
\begin{equation}\label{rut:7}
\sT Q\ni v\longmapsto -\langle p,v\rangle+\langle p_\alpha, \sT\zeta(v)\rangle\in\widetilde\sT Z_\alpha
\end{equation}
where $p$ is an arbitrary element of $C_\alpha$ over the point $\tau_Q(v)$.
\begin{definition}
  The {\it Routh reduction relation} for a vector field $X$ and value $\alpha\in\R$ is a symplectic relation $\mathcal{R}$ from $\sT^\ast \sT Q$ to $\sP\widetilde{\sT}Z_\alpha$ generated by the section (\ref{rut:7}).
\end{definition}

The Lagrangian generating object is then a family of sections of $\widetilde\sT Z_\alpha\rightarrow\sT Q_X$ parameterized by elements of $\sT Q$
\begin{equation}\label{rut:8}
\xymatrix{
\sT Q\ar[d]\ar[r] & \widetilde\sT Z_\alpha \\
\sT Q_X & },\quad
v\longmapsto L(v)-\langle p,v\rangle+\langle p_\alpha, \sT\zeta(v)\rangle
\end{equation}
The above generating family will be called {\it Routhian}.
Fibres of the bundle $\sT Q\rightarrow \sT Q_x$ on which Routhian is defined are two-dimensional. Simplifying the family means finding stationary points in the direction of $X$ in every tangent space $\sT_q Q$ and then along $\sd_\sT X$ trajectories. Differentiating in the direction of $X$ yields the equation
\begin{equation}\label{rut:9}
\langle \sd L(v), X^{\sv}\rangle=\alpha,
\end{equation}
while differentiating along $\sd_\sT X$ trajectories does not lead to any condition for $v$ in case $X$ is a Routh symmetry for $L$. The condition (\ref{rut:9}) can be expressed in terms of the Legendre map $\lambda$: $\langle\lambda(v), X\rangle=\alpha$. In another words, the set of solutions of (\ref{rut:9}) is equal to $\lambda^{-1}(C_\alpha)$. For symmetric Lagrangian we have the following proposition
\begin{prop}\label{prop:2}
  For Lagrangian satisfying (\ref{rr:1}) we have
\be\label{rr:3}
\Phi_t^\ast\circ\lambda=\lambda\circ\Phi_t
\ee
\end{prop}
\begin{proof}
Indeed, for any $v\in \sT_xQ$ and any $w\in \sT_{\varphi_t(x)}Q$ we have
\begin{multline*}
\langle\, \Phi^\ast_t(\lambda(v)), w\,\rangle= \langle\,\lambda(v), \Phi_{-t}(w) \,\rangle=
\frac{\sd}{\sd s}_{|s=0} L(v+s\Phi_{-t}(w))=\\
\frac{\sd}{\sd s}_{|s=0} L(v+\Phi_{-t}(sw))= \frac{\sd}{\sd s}_{|s=0} L(\Phi_t(v)+sw)= \langle\, \lambda(\Phi_t(v)), w\,\rangle
\end{multline*}
In the above calculation we have used the fact that $\Phi_t$ is linear in fibres and that $L$ is constant on the trajectories of $\sd_{\sT}X$. Note that (\ref{rr:3}) is true regardless of the regularity of the Lagrangian.
\end{proof}
Equation (\ref{rr:3}) means that $\lambda^{-1}(C_\alpha)$ is invariant with respect to the lifted $\R$-action. Note that even if the Lagrangian $L$ is hyperregular it may be not possible to simplify Routhian to just one section of $\widetilde\sT Z_\alpha\rightarrow \sT Q_X$ (see Example \ref{ex:2}).

\begin{corollary}\label{rut:10} In case the initial Lagrangian $L$ satisfies (\ref{rr:1}) the generating object for reduced dynamics is a family of sections of {\it av}\,-bundle $\widetilde \sT Z_\alpha\rightarrow \sT Q_X$ called Routhian, given by the formula (\ref{rut:8}). The value of the family does not depend on the choice of $p\in (C_\alpha)_{\tau_Q(v)}$. The condition for a stationary point is given by (\ref{rut:9}).
\end{corollary}

Let us now introduce appropriate coordinates and write the coordinate expression for Routhian generating family. Locally we can always choose such chart $(y,x^i)$ in $\mathcal{U}\subset Q$ that $X(y,x^i)=\frac{\partial}{\partial y}$. The $\R$-action is then given by $(t,y,x^i)\mapsto (y+t, x^i)$. The projection on $Q_X$ consists of omitting the first coordinate: $(y,x^i)\mapsto (x^i)$. Tangent and cotangent lifts of $X$ also have very simple expressions, namely
if $(y,x^i,\dot y,\dot x^j)$ and $(y,x^i,s,p_j)$ are adapted coordinates in $\sT Q$ and $\sT^\ast Q$ respectively, then
$$\sd_{\sT} X(y,x^i,\dot y,\dot x^j)=\frac{\partial}{\partial y}\qquad\text{and}\qquad
\sd_{\sT^\ast} X(y,x^i,s,p_j)=\frac{\partial}{\partial y}.$$
The level set $C_\alpha$ of $\imath_X$ is a subset of $\sT^\ast Q$ given by the condition $s=\alpha$.

Adapted coordinates in $Z_\alpha$ are introduced in the following way. First, we choose a reference section $\sigma_0$ of the {\it av}\,-bundle $Z_\alpha\rightarrow Q_X$ using coordinates in $Q$, namely $Q_X\ni (x^i)\mapsto \sigma_0(x)=[(0,x^i), 0]\in Z_\alpha$. The reference section allows us to identify $Z_\alpha$ with $Q_X\times\R$ and then use coordinates in $Q_X$. What we get is $[q,r]\mapsto (x^i(q), r-\alpha y(q))$.

We can then identify $\sP Z_\alpha$ with $\sT^\ast Q_X$ and us coordinates $(x^i, p_j)$ there. Similarly we identify $\widetilde\sT Z_\alpha$ with $\sT Q_X\times\R$ and use $(x^i,\dot x^j, \dot y)$ coordinates there.

The Routhian (\ref{rut:8}) now reads
$$\sT Q\ni (y, x^i, \dot y, \dot x^j)\longmapsto (x^i, \dot x^j, L(y, x^i, \dot y, \dot x^j)-\alpha\dot y-p_j\dot x^j+p_j\dot x^j)\in \widetilde\sT Z_\alpha.$$
We can see that it indeed does not depend on the choice of $p=(y,x^i,\alpha, p_j)\in C_\alpha$, because all summands containing $p_j$ cancel. What we finally have in coordinates is
\begin{equation}\label{rut:11}
\sT Q\ni (y, x^i, \dot y, \dot x^j)\longmapsto (x^i, \dot x^j, L(y, x^i, \dot y, \dot x^j)-\alpha\dot y)\in \widetilde\sT Z_\alpha.\end{equation}
For Lagrangians with symmetry, i.e. not depending on $y$ the family may be further simplified with the only parameter being $\dot y$. The last component in (\ref{rut:11}) is the coordinate form of Routhian present in the literature.

\begin{example}\label{ex:2}{\rm
Let $Q=\R^2$, $L(x,y,\dot x, \dot y)=\dot x\dot y-y^2$. Clearly $x$ is a cyclic variable, we can therefore perform Routh reduction with respect to the vector field $X=\frac{\partial}{\partial x}$. Situation here is very simple, in particular we have the structure of Cartesian product separating cyclic and non-cyclic variables. Instead of a Routhian family of sections we shall have just the Routhian family of functions. The formula (\ref{rut:8}) in this case reads
$$(x,y,\dot x,\dot y)\longmapsto  \dot x\dot y-y^2-\alpha \dot x\in\R.$$
Since we have no dependence on $x$, the family can be simplified to
\begin{equation}\label{ex:1}
(y,\dot x,\dot y)\longmapsto  \dot x\dot y-y^2-\alpha \dot x\in\R.
\end{equation}
It cannot however be simplified to one function over $\sT Q_X$, i.e. in variables $(y,\dot y)$ unless we accept constrained Lagrangian systems. Nevertheless, the family generates the dynamics in $\sT\sT^\ast Q_X$ which is given by a hamiltonian vector field. Indeed, (\ref{ex:1}) generates the following Lagrangian submanifold in $\sT^\ast \sT Q_X$
$$\{(y,\dot y, a, b)\in \sT^\ast\sT Q_X:\; \dot y=\alpha,\; a=-2 y\}.$$
Applying $\alpha_{Q_X}^{-1}$ we get a following submanifold of $\sT\sT^\ast Q_X$
$$\{(y,p,\dot y ,\dot p)\in \sT\sT^\ast Q_X:\; \dot y=\alpha,\; \dot p=-2 y\}$$
which clearly is the image of the vector field
$$X_h(y,p)=\alpha\frac{\partial}{\partial y}-2y\frac{\partial}{\partial p}.$$
The above field is a hamiltonian vector field for $h(y,p)=\alpha p+y^2$.
}\end{example}

\subsection{Routh reduction as Lagrangian reduction.}\label{sub:rutlag}

The purpose of this section is to show that Routh reduction relation $\mathcal{R}$ from $\sT^\ast\sT Q$ to $\sP\widetilde\sT Z_\alpha$ can be performed purely on Lagrangian side of the Tulczyjew triple. Indeed, we have the following proposition
\begin{prop}
The relation $\mathcal{R}$ is the affine phase lift of the {\it av}\,-bundle relation $\widetilde{\sT}(Q\times \R)\rightarrow \widetilde{\sT}Z_\alpha$ which is the reduced tangent to the projection $\zeta_\alpha$ (see equation (\ref{rut:1a})).
\end{prop}
\begin{proof}
Recall that $\zeta_\alpha:Q\times \R\rightarrow Z_\alpha$ is a natural projection given by the equivalence relation (\ref{rut:1b}). As an {\it av}\,-bundle morphism it covers the map $\zeta: Q\rightarrow Q_X$. It follows that the tangent map $\sT\zeta_\alpha$ is given by the tangent equivalence relation in $\sT(Q\times\R)\simeq \sT Q\times \sT\R$:
\begin{multline*}
(v_1,r_1,\dot r_1)\sim_{\sT}(v_2,r_2,\dot r_2)\quad\text{if}\quad v_2=\sT\varphi_s(v_1)+\dot sX(\tau_Q(v_2)),\\ r_2=r_1+\alpha s, \; \dot r_2=\dot r_1+\alpha\dot s.
\end{multline*}
for some $s,\dot s\in \R$. Similarly, the reduced tangent relation $\widetilde{\sT}\zeta_\alpha: \widetilde{\sT}(Q\times\R)\rightarrow\widetilde{\sT}Z_\alpha $ is given by the equivalence relation in
$\widetilde{\sT}(Q\times\R)=\sT Q\times \R$:
$$(v_1,t_1)\sim_{\widetilde\sT}(v_2, t_2)\quad\text{if}\quad v_2=\sT\varphi_s(v_1)+\dot s X(\tau_Q(v_2)), \; t_2=t_1+\alpha\dot s$$
for some $s,\dot s\in \R$. Let us note that $\widetilde{\sT}\zeta_\alpha$ is a morphism of {\it av}\,-bundles.

The graph of $\widetilde{\sT}\zeta_\alpha$ is a section of $\widetilde{\sT}Z_\alpha\ominus \widetilde{\sT}(Q\times\R)$ over the graph of projection $\zeta$. We shall use the obvious identification of $\widetilde{\sT}Z_\alpha\ominus \widetilde{\sT}(Q\times\R)$ with $\widetilde{\sT}Z_\alpha\times \sT Q$, according to which we identify an element $(\tilde{v},w)\in\widetilde{\sT}Z_\alpha\times \sT Q$ with the class $[\tilde{v}, (w,0)]$ in $\widetilde{\sT}Z_\alpha\ominus \widetilde{\sT}(Q\times\R)$. On the other hand $\mathrm{graph}(\widetilde{\sT}\zeta_\alpha)$ consists of elements $[\widetilde{\sT}\zeta_\alpha(w,t), (w,t)]=[\widetilde{\sT}\zeta_\alpha(w,0), (w,0)]$. As a subset of $\widetilde{\sT}Z_\alpha\times \sT Q$ the graph of $\widetilde{\sT}\zeta_\alpha$ consists of pairs $(\widetilde{\sT}\zeta_\alpha(w,0), w)$.

Let us now fix $w\in \sT Q$ and choose any $p\in C_\alpha$ such that $\pi_Q(p)=\tau_Q(w)$. We denote by $p_\alpha$ the element of $P_\alpha\simeq \sP Z_\alpha$ corresponding to $p$.
The affine covector $p_\alpha$ evaluated on $\sT\zeta(w)$ gives an element $\langle\, p_\alpha, \, \sT\zeta(w)\, \rangle\in \widetilde\sT\zeta_\alpha$ which is represented by a pair
$(w,\langle p,w\rangle)\in\sT Q\times \R$ with respect to the relation $\sim_{\widetilde{\sT}}$. It means that the element $\langle\, p_\alpha, \, \sT\zeta(w)\, \rangle -\langle w,p\rangle$ is represented by $(w,\langle w,p\rangle)-\langle p,w\rangle=(w,0)$. We have shown then that
$$\sT\zeta_\alpha(w,0)=\langle\, p_\alpha, \, \sT\zeta(w)\, \rangle -\langle w,p\rangle$$
for any $p\in C_\alpha$ over the appropriate point in $Q$. It is then true that the generating object (\ref{rut:7}) is the same as graph of $\widetilde\sT\zeta_\alpha$.
\end{proof}

The above proposition shows the main difference between Hamiltonian and Lagrangian reductions. Hamiltonian reduction is generated by the zero function on the graph of the underlying relation while the corresponding Lagrangian reduction is generated by the nontrivial section over the graph of the underlying relation. Lagrangian reduction is the reduction `with values'. Since $\widetilde\sT Z_\alpha$ is a not trivial {\it av}\,-bundle over $\sT Q_X$, it is not possible to perform the reduction on the Lagrangian side without the
use of the language the differential calculus of affine values.

\subsection{Examples}\label{sub:ex}

\begin{example}{\rm
A gauge independent formulation of dynamics of a relativistic charged particle can be obtained by extending the configuration space-time $Q$ to the total space $Z$ of a principal $(\R, +)$ fibration $\tau: Z\rightarrow Q$. The electromagnetic potential is then a principal bundle connection form $A$ on $Z$. The Lagrangian of a particle with mass $m$ and charge $e$ is the following gauge invariant function on $\sT Z$
\begin{equation}\label{ex:3}
L(w)=m\sqrt{g(\sT\tau(w), \sT\tau(w))}+e\langle A, w\rangle,
\end{equation}
where $g$ is the Minkowski metric on $Q$. The Lagrangian is invariant with respect to the lift of the action of $(\R,+)$ to $\sT Z$. It is, however, singular and implies constraints $\langle p,v\rangle=e$ in the phase space $\sT^\ast Z$. The Routh reduction makes sense for the parameter $\alpha=e$ only. Indeed, in local coordinates the Routh generating family is the following
\begin{equation}\label{ex:4}
\xymatrix{\widetilde{\sT} Z\ni (x^i, \dot y, \dot x^j)\ar@{|->}[r]\ar[d] & m\sqrt{g_{ij}\dot x^i\dot x^j}+e\dot y-eA_i\dot x^i-\alpha \dot y\in \widetilde{\sT}Z_\alpha\\
\sT Q & }
\end{equation}
The stationary point with respect to $\dot y$ must satisfy $e-\alpha=0$. For $\alpha=e$ we obtain the Routhian which is a section of $\widetilde{\sT}Z_e\rightarrow\sT Q$. In local coordinates the Routhian reads
\begin{equation}\label{ex:5}
(x^i,\dot x^j)\longmapsto (x^i, \dot x^j, m\sqrt{g_{ij}\dot x^i\dot x^j}-eA_i\dot x^i).
\end{equation}
We can provide also the coordinate free interpretation of Routhian. The electromagnetic potential can be represented by an affine one-form which is a section of the affine phase bundle $\sP Z\rightarrow Q$. We shall denote it by $A$ since it contains the same information as an appropriate connection form on $Z$. This form induces an affine one-form $A_\alpha$ i.e. a section of $\sP Z_\alpha\rightarrow Q$. The Routhian (\ref{ex:5}) can be written in the following form
\begin{equation}\label{ex:6}
\sT Q\ni v\longmapsto m\sqrt{g(v,v)}+\langle A_e,v\rangle\in \widetilde{\sT} Z_e.
\end{equation}
The above section is the gauge independent Lagrangian introduced in \cite{TU}.
}\end{example}

\begin{example}{\rm
The Jacobi variational principle for a mechanical system with Lagrangian $L:\sT Q\rightarrow \R$ is a variational principle for the images of motions, i.e., one-dimensional submanifolds of $Q$. A Lagrangian $L_J$ for this principle should therefore be homogeneous, which leads to constraints in the phase space $\sT^\ast Q$. The Hamiltonian generating object for Jacobi system is function equal to zero on the constraints. The obvious choice for the constraints is the level set of the Hamiltonian $H:\sT^\ast Q\rightarrow \R$ corresponding to the initial Lagrangian $L$, i.e.
$$C_E=\{p\in\sT^\ast Q:\; H(p)=E\}.$$
To obtain the Lagrangian $L_J$ of the Jacobi system we perform the Legendre transformation as described in \cite{TU2} in the oposite direction, looking for a Lagrangian corresponding to the Hamiltonian generating object (function equal to zero on $C_E$). The appropriate Lagrangian generating family is
\begin{equation}\label{ex:7}
\xymatrix{C_E\times_Q\sT Q \ar[r]\ar[d] & \R: & (p,v) \ar@{|->}[r] & \langle p,v\rangle \\
\sT Q & }
\end{equation}
To illustrate the procedure we will reduce the above family in the special case of standard Lagrangian $L(v)=\frac12g(v,v)-V(\tau_Q(v))$. The Hamiltonian is then
$H(p)=\frac12g(p,p)+V(\pi_Q(p))$, in local coordinates
$$L(x^i,\dot x^j)=\frac12g_{ij}\dot x^i\dot x^j-V(x^i),\qquad H(x^i,p_j)=\frac12g^{ij}p_ip_j+V(x^i).$$
We find a stationary point of (\ref{ex:7}) using the Lagrange multiplier method. The conditions are then
$$\lambda\frac{\partial H}{\partial p_i}=\dot x^i, \quad H(x^i, p_j)=E.$$
Solving for $\lambda$ and $p_j$ we get
$$\lambda=\pm \sqrt{g_{ij}\dot x^i\dot x^j}\frac{1}{\sqrt{2(E-V(x^i))}}, \qquad p_j=\frac{1}{\lambda} g_{ij}\dot x^i.$$
Finally the reduced Lagrangian generating object reads
\begin{equation}\label{ex:8}
  L_J(v)=\pm\sqrt{2}\sqrt{g(v,v)}\sqrt{E-V(\tau_Q(v))}.
\end{equation}
The choice of the sign is the choice of the orientation of trajectories.

Another approach, proposed by Ba\.za\'nski in \cite{Ba}, makes use of Routh reduction. The first step is to make the initial system homogeneous by adding an extra parameter, i.e. we consider a Lagrangian system on $Q\times\R$ with the Lagrangian
$$L_h: \sT(Q\times \R)\ni (v,s,\dot s)\longmapsto \dot sL\left(\frac{v}{\dot s}\right)\in\R.$$
The function $L_h$ is clearly homogeneous and $s$ is a cyclic variable. Then we find the constraints in $\sT^\ast(Q\times \R)$ due to homogeneity of $L_h$. Denoting by $\lambda:\sT Q\rightarrow \sT^\ast Q$ the Legendre map for the initial Lagrangian we get that $(p,s,\tau)\in \sT^\ast(Q\times \R)$ must satisfy
$$p=\lambda\left(\frac{v}{\dot s}\right), \quad\tau=L\left(\frac{v}{\dot s}\right)-\frac{1}{\dot s}\langle\lambda\left(\frac{v}{\dot s}\right), v\rangle$$
For a hyperregular Lagrangian $L$ we get $\tau=-H(p)$ where $-H$ is the Hamiltonian corresponding to $L$.

The second step is to perform the Routh reduction with the Routh parameter $-E$. We get the Routhian looking for the stationary value with respect to $\dot s$ of the family
$$\xymatrix{\sT Q\times \R\ar[d]\ar[r] & \R : & (v,\dot s)\ar@{|->}[r] & L_h(v,\dot s)+E\dot s\\
\sT Q}$$
The stationary point has to satisfy the condition
$$L\left(\frac{v}{\dot s}\right)-\frac{1}{\dot s}\langle\lambda\left(\frac{v}{\dot s}\right), v\rangle +E=0.$$
We solve this equation for the mechanical Lagrangian getting
$$\dot s=\pm \frac{\sqrt{g(v,v)}}{\sqrt{2(E-V(\tau_Q(v)))}}.$$
The Routhian is then
$$R(v)=\pm  \sqrt{g(v,v)}\sqrt{2}\sqrt{(E-V(\tau_Q(v)))}$$
which is precisely $L_J$ from (\ref{ex:8}).

}\end{example}

\subsection{Reduction to mechanics on an algebroid}\label{sub:alg}

To put Routh reduction in some perspective let us now consider for a while another reduction of the system with symmetry (\ref{rr:1}), namely reduction to mechanics on an algebroid.
We can consider $\zeta: Q\rightarrow Q_X$ as a principal bundle with structure group $(\R, +)$. The condition (\ref{rr:1}) makes it possible to describe the system as Lagrangian system on Atiyah algebroid $\widetilde\sT Q$.

The condition (\ref{rr:1}) means that $L$ is constant on trajectories of $\sd_{\sT}X$, it can be then regarded as the pull-back of a function $\ell$ on $\widetilde\sT Q$. The appropriate diagram is then the following
\be\label{rr:4} \xymatrix@C-15pt@R-5pt{
 & & & & \mathcal{D_\ell}\ar@{ (->}[d] & & & & \\
 & \sT^\ast \widetilde\sT Q^\ast \ar[rrr]^{\tilde\Lambda}  \ar[dr] \ar[ddl]
 & & & \sT \widetilde\sT Q^\ast \ar[dr]\ar[ddl]
 & & & \sT^\ast \widetilde\sT Q\ar[dr]_{\pi_{\sT Q_X}}\ar[ddl]\ar[lll]_{\varepsilon} & \\
 & & \widetilde\sT Q\ar@{.}[rrr]\ar@{.}[ddl] & & & \sT Q\ar@{.}[rrr]\ar@{.}[ddl]
 & & & \widetilde\sT Q \ar@{.}[ddl]\ar@/_1pc/[ul]_{\sd \ell}\ar[dll]_{\lambda}\\
 \widetilde\sT Q^\ast\ar@{.}[rrr]\ar@{.}[dr] \ar@/^1pc/[uur]^{\sd h} & & & \widetilde\sT Q^\ast\ar@{.}[rrr]\ar@{.}[dr]
 & & & \widetilde\sT Q^\ast \ar@{.}[dr] & &  \\
 & Q_X\ar@{.}[rrr]& & & Q_X\ar@{.}[rrr]& & & Q_X &
} \ee

Let us consider relations between (\ref{not:15}) and (\ref{rr:4}). On Lagrangian side, $\sT^\ast \sT Q$ is reduced with respect to the coisotropic submanifold $K=\langle \sd_\sT X\rangle^\circ$. Since $\sd L(\sT Q)\subset K$ and $\sd L(\sT Q)$ is a Lagrangian submanifold it is composed of leaves of characteristic foliation of $K$ (which are trajectories of $\sd_{\sT} X$) and reduces to $\sd\ell(\widetilde\sT Q)\subset \sT^\ast \widetilde\sT Q$.

Whatever happens on the Lagrangian side of the triple has consequences also in the middle i.e. for the dynamics, since $\alpha_Q$ is a symplectomorphism. We get that the dynamics $\mathcal{D}=\alpha_Q^{-1}(\sd L(\sT Q))$ is contained in the coisotropic submanifold $\alpha_Q^{-1}(K)$ of codimension $1$ and again is composed of leaves of characteristic foliation of $K$. Those leaves happen to be trajectories of $\sd_{\sT^\ast} X$. Indeed, we have a following proposition
\begin{prop}\label{prop:1}
Let $\imath_X$ denote the fiberwise linear function on $\sT^\ast Q$  corresponding to the vector field $X$ on $Q$, i.e. $\imath_X(p)=\langle\, p, X(\pi_Q(p))\, \rangle$. Then
$$\alpha_Q^{-1}(K)=\{ v\in\sT\sT^\ast Q:\;\; \langle \sd \imath_x, v\rangle=0\,\}.$$
\end{prop}
\noindent{\bf Proof:} The proof is based on simple calculation. By definition $v\in \alpha_Q^{-1}(K)$ if and only if $\langle\, \alpha_Q(v),\sd_{\sT}X \,\rangle=0$. Let us start from there:
$$0=\langle\, \alpha_Q(v),\sd_{\sT}X \,\rangle=\langle\!\langle\, v, \kappa_Q(\sd_{\sT}X)\,\rangle\!\rangle,$$
Using definition of $\sd_{\sT}X$ we get
$$\langle\!\langle\, v, \sT X(\sT\pi_Q(v))\,\rangle\!\rangle=0.$$
Let us now take any curve $\gamma: I\rightarrow \sT^\ast Q$, such that $\dot\gamma(0)=v$, then it follows from the definition of the pairing that
\begin{equation*}
0=\,\langle\!\langle v, \sT X(\sT\pi_Q(v))\,\rangle\!\rangle=\frac{\sd}{\sd t}_{|t=0}\langle\, \gamma(t), X(\pi_Q(\gamma(t))\,\rangle=
\frac{\sd}{\sd t}_{|t=0}\imath_X(\gamma(t))=\langle\,\sd\imath_X,v\,\rangle.
\end{equation*}
$\Box$

\noindent The above proposition shows that the dynamics $\mathcal{D}$ reduces to the dynamics $\mathcal{D}_\ell$ generated in $\sT\widetilde \sT^\ast Q$ by $\ell$. Similar reduction can be done on Hamiltonian side.

\begin{prop}\label{prop:5} If Lagrangian satisfying (\ref{rr:1}) is hyperregular then the appropriate Hamiltonian is constant on trajectories of $\sd_{\sT^\ast} X$, i.e. satisfies
\be\label{rr:2}
\sd_{\sT^\ast}X(H)=0,
\ee
\end{prop}

\noindent{\bf Proof:} Let $\lambda:\sT Q\rightarrow \sT^\ast Q$ be the Legendre map associated to $L$. By definition $\lambda(v)=\xi(\sd L(v))$. Alternatively, we can write
$$\langle\, \lambda(v), w\rangle=\frac{\sd}{\sd s}_{|s=0} L(v+sw).$$
Recall that Lagrangian satisfying (\ref{rr:1}) we have $\Phi_t^\ast\circ\lambda=\lambda\circ\Phi_t$ (Proposition \ref{prop:2}). For hyperregular Lagrangians the Legendre map is  invertible and then
\begin{multline*}H(\Phi^\ast_t(p))=L(\lambda^{-1}(\Phi^\ast_t(p)))-\langle \Phi^\ast_t(p)\, ,\,\lambda^{-1}(\Phi^\ast_t(p))\rangle=\\
L(\Phi_t(\lambda^{-1}(p)))-\langle \Phi^\ast_t(p)\, ,\,\Phi_t(\lambda^{-1}(p))\rangle= \\
L(\lambda^{-1}(p))-\langle \Phi^\ast_{-t}\Phi^\ast_t(p)\, ,\,(\lambda^{-1}(p))\rangle=H(p).
\end{multline*}
$\Box$

For Lagrangians that are not hyperregular dynamics is not generated by Hamiltonian function but by family of functions. There is a standard choice of this family, where functions are parameterized by velocities:
$$ F_H:\sT^\ast Q\times_Q\sT Q\longrightarrow \R, \quad F_H(p,v)=L(v)-\langle p,v\rangle.$$
This family is also invariant with respect to $(\R, +)$ action when we use cotangent and tangent lifts in $\sT^\ast Q$ and $\sT Q$ respectively.

Summarizing, mechanical system on $Q$ invariant with respect to $X$ in a sense of (\ref{rr:1}) reduces to the mechanical system on $\widetilde\sT Q$ with phase space $\widetilde\sT^\ast Q$ and phase dynamics being a subset of $\sT \widetilde\sT^\ast Q$. Comparing algebroid reduction with Routh reduction we see that the first one is a step in between the original system and the system `Routh-reduced'. If $n=\dim Q$ then the phase space for the original system is $2n$-dimensional, after algebroid reduction we get ($2n-1$)-dimensional phase space and finally $\sP Z_\alpha$ is ($2n-2$)-dimensional. There is a counterpart of this ``intermediate'' reduction on the Hamiltonian side. It is the symplectic reduction with respect to the coisotropic submanifold $\langle \sd_\sT X\rangle^0\subset \sT^\ast\sT Q$. The reduced manifold is canonically isomorphic to $\sT^\ast\widetilde{\sT}Q$.

\section{(Not really a) generalization}\label{sec:gen}

Let us observe that if a vector field $Y$ on $\sT Q$ is a symmetry of a Lagrangian so is the field $fY$ for any function $f$ on $\sT Q$.  What matters here is the distribution spanned by $Y$, not the field itself. For Routh reduction, however, we use special symmetries i.e. symmetries that are in some sense lifted from the space of positions. Let us then investigate the properties of the distribution $\Delta_X$ spanned by $\sd_{\sT}X$ in more detail to determine assumptions that should be made for a distribution representing such symmetries.

\begin{prop}\label{gen:1} The distribution $\Delta_X$ is a double vector subbundle of $\sT\sT Q$.
\end{prop}

\noindent{\bf Proof:}  Recall that a double vector subbundle is a submanifold of the total space of the double vector bundle, which is a subbundle of both vector bundle structures, possibly supported on submanifolds. For nonvanishing vector field $X$ its lift $\sd_{\sT}X$ is also nonvanishing, so the distribution $\Delta_X$ is one-dimensional subbundle of the bundle $\tau_{\sT Q}$ and $(2n+1)$-dimensional submanifold of $\sT\sT Q$. $\Delta_X$ projects by $\sT\tau_Q$ on the submanifold
$\Delta_0$ which is the distribution on $Q$ spanned by $X$. Let us fix a point $w$ in $\Delta_0$. There exists $a\in\R$ such that $w=aX(q)$. The intersection  $\Delta_X\cap(\sT\tau_Q)^{-1}(w)$ of the distribution with the fibre of tangent projection over $w$ equals  $\kappa_Q(\sT(aX)(\sT Q))$ which is a vector subspace, since $\sT(aX)$ is a linear map and $\kappa_Q$ is a double vector bundle morphism. The distibution $\Delta_X$ is then a subbundle of the tangent projection as well.

Double vector subbundle is a double vector bundle itself, so we can draw the diagram showing both projections and the core of $\Delta_X$.
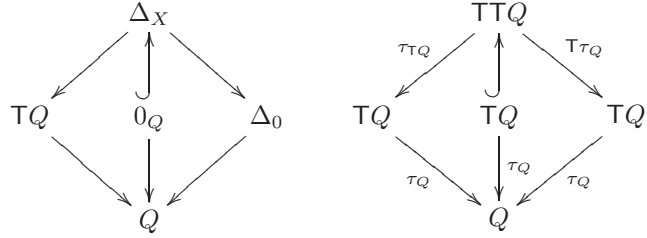
\begin{equation}\label{gen:2}
\xymatrix{ & \Delta_X\ar[dl]\ar[dr] & \\
\sT Q\ar[dr] & 0_Q\ar[d]\ar@{ (->}[u] & \Delta_0\ar[dl] \\
& Q &
}\qquad
\xymatrix{ & \sT\sT Q\ar[dl]_{\tau_{\sT Q}}\ar[dr]^{\sT\tau_Q} & \\
\sT Q\ar[dr]_{\tau_Q} & \sT Q\ar[d]^{\tau_Q}\ar@{ (->}[u] & \sT Q\ar[dl]^{\tau_Q} \\
& Q &
}
\end{equation}
$\Box$

\noindent It will be useful to have the following definition
\begin{definition}\label{gen:4}
A distribution $\Delta$ on the total space $E$ of a vector bundle $E\rightarrow M$ is called linear if it is a double vector subbundle of $\sT E$.
\end{definition}

\noindent We could now reformulate proposition \ref{gen:1} saying that $\Delta_X$ is a linear distribution on $\sT Q$.

Now we come back to the properties of $\Delta_X$. Let $\Delta_X^+$ be the anihilator of $\Delta_X$ with respect to the tangent structure, i.e.
$$\Delta_X^+=\{w\in \sT\sT^\ast Q: \; \langle\!\langle w,v\rangle\!\rangle=0 \text{ for } v\in\Delta_X, \sT\tau_Q(v)=\sT\pi_Q(w)\,\}$$

\begin{prop}\label{gen:3}
$\Delta_X^+$ is a distribution on $\sT^\ast Q$ spanned by $\sd_{\sT^\ast} X$. As a submanifold it is coisotropic with respect to $\sd_\sT\omega_Q$.
\end{prop}

\noindent{\bf Proof:} Since $\Delta_X$ is a linear distribution so is $\Delta_X^+$. The appropriate diagram reads
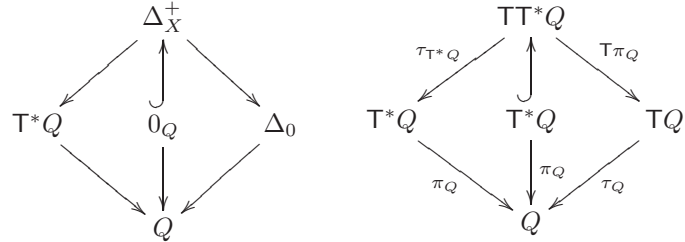
\begin{equation}\label{gen:7}
\xymatrix{ & \Delta_X^+\ar[dl]\ar[dr] & \\
\sT^\ast Q\ar[dr] & 0_Q\ar[d]\ar@{ (->}[u] & \Delta_0\ar[dl] \\
& Q &
}\qquad
\xymatrix{ & \sT\sT^\ast Q\ar[dl]_{\tau_{\sT^\ast Q}}\ar[dr]^{\sT\pi_Q} & \\
\sT^\ast Q\ar[dr]_{\pi_Q} & \sT^\ast Q\ar[d]^{\pi_Q}\ar@{ (->}[u] & \sT Q\ar[dl]^{\tau_Q} \\
& Q &
}
\end{equation}
The projection of $\Delta_X^+$ by $\tau_{\sT^\ast Q}$ is the anihilator of the core of $\Delta_X$, so it is indeed the whole $\sT^\ast Q$. The dimension of the fibre of $\Delta_X$ over fixed $w\in \Delta_0$ equals $n$, so its anihilator is also $n$ dimensional subspace in the fibre of $\sT\pi_Q$ over $w$. Since $\Delta_0$ as a submanifold is $(n+1)$-dimensional, $\Delta_X^+$ is $(2n+1)$-dimensional submanifold in $\sT\sT^\ast Q$. This shows that it is one-dimensional distribution on $\sT^\ast Q$.

For a covector $p$ and vector $v$ over the same point in $Q$ we take integral curves of $\sd_{\sT^\ast} X$ and $\sd_{\sT} X$ respectively:
$$t\longmapsto \eta(t)=\sT^\ast\phi_t(p)\in\sT^\ast Q\qquad t\longmapsto \gamma(t)=\sT\phi_t(v)\in\sT Q.$$
Curves $\eta$ and $\gamma$ cover the same curve in $Q$, we can then calculate the pairing of $\eta(t)$ and $\gamma(t)$. It follows from the definition of flows of both fields that the function
$$t\longmapsto\langle\eta(t), \gamma(t)\rangle=\langle p,v\rangle$$
is constant. We get that
$$\langle\!\langle\sd_{\sT^\ast} X(p),\sd_{\sT} X(v) \rangle\!\rangle=\frac{d}{dt}\langle\eta(\cdot), \gamma(\cdot)\rangle_{|t=0}=0$$
It follows that the image of $\sd_{\sT^\ast} X$ lies in $\Delta_X^+$. $\sd_{\sT^\ast} X$ is a nonvanishing vector field, so it spans the distribution.

Let us now consider a closed, nonvanishing smooth form $\varphi$ on a manifold $M$ an define $N=\{p\in\sT^\ast M: \, p=a\varphi(x), a\in \R\}$. $N$ is $(m+1)$-submanifold of
$\sT^\ast M$ for $m=\dim M$. We will show that $N$ is coisotropic. Let $p=a\varphi(x)$. The tangent space $\sT_p N$ contains $\sT_p(a\varphi(M))$ as a subspace. The form $a\varphi$ is closed, so $a\varphi(M)$ is Lagrangian in $\sT^\ast M$. It follows that
$$(\sT_p N)^\S\subset (\sT_p(a\varphi(M)))^\S=\sT_p(a\varphi(M))\subset\sT_p N,$$
which means that $N$ is coisotropic. Taking $M=\sT^\ast Q$ and $\varphi=\sd\imath_X$ we conclude that the submanifold `spanned' by $\sd\imath_X$ is coisotropic. The distribution $\Delta_X^+$ is the image of this submanifold by the symplectomorphism $\beta_M$ so it is coisotropic as well.$\Box$

Following propositions \ref{gen:1} and \ref{gen:3} we should assume for the distribution $\Delta$ on $\sT Q$ generalizing $\sd_{\sT}X$ to be one-dimensional, linear,
and such that $\Delta^+$ is coisotropic with respect to $\sd_\sT\omega_Q$. Additional assumption is that $\Delta$ is not vertical, i.e. not contained in $\sV\sT Q$. It means that it projects on a distribution $\Delta_0$ on $Q$. Let $\Delta^\S$ be the symplectic anihilator of $\Delta^+$ (with respect to $\omega_Q$).

\begin{prop}\label{gen:5} With the above assumptions for $\Delta$, the distribution $\Delta^\S$ is linear, integrable and of codimension $1$.
\end{prop}

\noindent{\bf Proof:} By definition $\Delta^\S=\beta_M(\Delta^+)^\circ$. The anihilator is taken with respect to the vector bundle structure $\pi_{\sT^\ast M}:\sT^\ast\sT^\ast M\rightarrow \sT^\ast M$. Since $\beta_M(\Delta^+)$ is coisotropic, $\Delta^\S$ is integrable. The distribution $\Delta^+$ is one dimensional, therefore its symplectic anihilator is of codimension 1. Linearity of $\Delta^\S$ follows from the fact that $\beta_M$ is a double vector bundle morphism and that anihilator of a double vector subbundle is a double vector subbundle of the dual. $\Box$

Finally, we have the following proposition

\begin{prop}\label{gen:6} Let $C$ be an integral submanifold of $\Delta^\S$ not intersecting zero-section. Locally, there exists unique linear function $h$ on
$\sT^\ast Q$ such that $h_{|C}=1$ and $\Delta=\Delta_X$ for $X$ being the vector field on $Q$ corresponding to $h$.
\end{prop}

\noindent{\bf Proof:} We assume that $\Delta$ is a one-dimensional, linear and not vertical distribution on $\sT Q$ such that $\Delta^+$ is coisotropic. As previously by $\Delta_0$ we denote $\sT\tau_Q(\Delta)$ which is an one-dimensional distribution on $Q$. We have the following diagrams for $\Delta^+\subset \sT\sT^\ast Q$ and
$\beta_M(\Delta^+)\subset \sT^\ast\sT^\ast Q$
$$\xymatrix{ & \Delta_X^+\ar[dl]\ar[dr] & \\
\sT^\ast Q\ar[dr] & 0_Q\ar[d]\ar@{ (->}[u] & \Delta_0\ar[dl] \\
& Q &
}\qquad
\xymatrix{ & \beta_M(\Delta_X^+)\ar[dl]\ar[dr] & \\
\sT^\ast Q\ar[dr] & 0_Q\ar[d]\ar@{ (->}[u] & \Delta_0\ar[dl] \\
& Q &
}$$
To get $\Delta^\S$ we take the anihilator with respect to the left-hand-side projection. Following the rules for double vector bundles we get the diagram for $\Delta^\S\subset\sT\sT^\ast Q$
$$\xymatrix{ & \Delta^\S\ar[dl]\ar[dr] & \\
\sT^\ast Q\ar[dr] & (\Delta_0)^\circ\ar[d]\ar@{ (->}[u] & \sT Q\ar[dl] \\
& Q &
}$$
Let us discuss the vertical part of $\Delta^\S$, i.e. $\sV\sT^\ast Q\cap \Delta^\S$. Since $\sT^\ast Q\rightarrow Q$ is a vector bundle, its vertical tangent bundle has some additional internal structure, i.e. $\sV\sT^\ast Q\simeq \sT^\ast Q\times_Q\sT^\ast Q$. It is easy to see now that $\sV\sT^\ast Q\cap \Delta^\S\simeq \sT^\ast Q\times_Q(\Delta_0)^\circ$. It follows that the intersection of the leaf $C$ of the foliation defined by $\Delta^\S$ with $\sT^\ast_x Q$ for a given $x\in Q$ is an affine subspace modelled on $(\Delta_0)^\circ$. If the intersection contains zero covector it should be $(\Delta_0)^\circ$ itself. Let us fix a leaf $C$ that does not contain zero. It is then the codimension one affine subbundle of the cotangent bundle modeled on the anihilator of $\Delta_0$. Locally there exist a linear function $h$ on $\sT^\ast Q$ such that $h_{|C}=1$. The differential $dh$ vanishes on $\Delta^\S$, so the hamiltonian vector field generated by $X_h$ belongs to $\Delta^+$. A linear function on $\sT^\ast Q$ defines a vector field $X$ on $Q$ such that $h(p)=\langle p,X(\pi_Q(p))\rangle$. We know that then $X_h=\sd_{\sT^\ast}X$. Since $h$ is not zero then $X$ is nonvanishing and $\sd_{\sT^\ast}X$ is also nonvanishing. W know that $\Delta^+$is one-dimensional, therefore it is spanned by $\sd_{\sT^\ast}X$. It is clear now that $\Delta$ is spanned by $\sd_{\sT}X$. $\Box$

From the Proposition \ref{gen:6} we see that replacing a vector field $X$ with the distribution $\Delta$ as an object encoding symmetry is not really a generalization,
since every one-dimensional, non-vertical, linear distribution with the assumption that $\Delta^+$ is coisotropic is spanned by the lift of a vector field. In the proof above we have chosen a function $h$ such that assumes value $1$ on $C$. We could of course have chosen any non-zero value rescaling our vector field $X$, so in fact the vector field is  determined up to the multiplication by constant. This means that with $\Delta$ as our main object we are almost back in the old situation with $X$ but without the distinguished parameterization of leaves of the foliation given by $\Delta^\S$.

For the sake of completeness let us translate the assumption that $\Delta^+$ is coisotropic into the language of tangent bundle, i.e. using only the structure of $\sT\sT Q$.

\begin{prop}\label{gen:8} Let $\Delta$ be a linear one-dimensional and not vertical distribution on $\sT Q$. $\Delta^+$ is coisotropic if and only if $\kappa_Q(\Delta)$ is an integrable distribution on $\Delta_0$.
\end{prop}

\noindent{\bf Proof:} If $\Delta^+$ is coisotropic then we are in the situation described in  Proposition \ref{gen:6} which means that $\Delta$ is spanned by the tangent lift $\sd_{\sT} X$ of some non-vanishing vector field $X$ on $Q$. This vector field spans $\Delta_0$. First lets check if $\kappa_Q(\Delta)$ is tangent do $\Delta_0$ at all. Any element of $\Delta$ is of the form $w=a\sd_{\sT} X(v)$ for some $a\in\R$ and $v\in\sT Q$. Lifting a vector field is a linear process so we can write $w=\sd_{\sT} aX(v)$ as well, i.e. first multiply a vector field by a number and then lift. By definition of the lift we have that $w=\kappa_Q(\sT(aX)(v))$. Applying $\kappa_Q$ again we get $\kappa_Q(w)=\sT(aX)(v)$ which means that $\kappa_Q(w)$ is tangent to the image of the vector field $aX$ at the point $aX(\tau_Q(v))$. Since $aX$ is a section of $\Delta_0$, $\kappa_Q(w)$ is definitely tangent to $\Delta_0$. Moreover, the distribution $\kappa_Q(\Delta)$ at $aX(x)$ is just equal $\sT (aX)(\sT_xQ)$. We see then that $\kappa_Q(\Delta)$ treated as a distribution on $\Delta_0$ is integrable. The integral submanifolds are locally parameterized by real numbers and equal to the images of $aX$.

Now we assume that $\kappa_Q(\Delta)$ is integrable on $\Delta_0$. It means that its anihilator is coisotropic in $\sT^\ast \Delta_0$. The anihilator $\kappa_Q(\Delta)^\circ$ of $\kappa_Q(\Delta)$ in $\sT^\ast\sT Q$ is also coisotropic as a pre-image of a coisotropic submanifold by symplectic relation. It is easy to check that
$\alpha_Q(\Delta^+)=\kappa_Q(\Delta)^\circ$. Indeed, a covector $\varphi\in\sT^\ast\sT Q$ belongs to $\kappa_Q(\Delta)^\circ$ if it projects on $v_0\in\Delta_0$ and for all $v\in \Delta$ such that $\sT\tau_Q(v)=v_0$ we have $\langle \varphi,\kappa_Q(v)\rangle=0$. Using the definition of $\alpha_Q$ we can say that
The map $\alpha_Q$ is a diffeomorphism, therefore we can write that $\alpha_Q(w)\in \kappa_Q(\Delta)^\circ$ if and only if for all $v$ as before
$$0=\langle \alpha_Q(w),\kappa_Q(v)\rangle=\langle\!\langle w,\kappa_Q\circ\kappa_Q(v)\rangle\!\rangle=
\langle\!\langle w,v\rangle\!\rangle$$
The condition $\langle\!\langle w,v\rangle\!\rangle=0$ means that $w\in \Delta^+$. Since $\alpha_Q$ is a symplectomorphism we conclude that $\Delta^+$ is coisotropic. $\Box$

Thus we have, associated with $\Delta$, two local foliations: integral submanifolds of $\Delta^+{}^\S\subset \sT\sT^\ast Q$ and integral submanifolds of $\kappa_Q(\Delta)\subset\sT L$. The first one is a family $C_a$ of affine subbundles of $\sT^\ast Q$. The second is a family $\{X_b\}$ of vector fields on $Q$. Two vector fields of the family differ by the multiplicative constant. Any vector field from the family $\{X_b\}$ defines a function on $\sT^\ast Q$ with $\{C_a\}$ as level sets. The symplectic reduction $P_a$ of $C_a$ is an affine bundle modelled on $\sT^ast Q_X$. For given $X_b$, $P_a$ can be identified with $\sP Z_{a,b}$ where $Z_{a,b}$ is defined by an equivalence relation in $Q\times \R$
$$(q,t)\sim (\varphi_{X_b,s}(q), t+\alpha_{a,b}s)$$
where
$\alpha_{a,b}=\langle p, X_b\rangle$ for $p\in C_a$. Now, we replace $X_b$ by $X_{b'}=cX_{b}$ and we get new equivalence relation
$$(q,t)\sim (\varphi_{X_{b'},s}(q), t+\alpha_{a,b'}s)$$
Since $X_{b'}=cX_{b}$ we have $\varphi_{X_{b'},s}=\varphi_{X_{b},cs}$ and the above relation assumes the form
$$(q,t)\sim (\varphi_{X_{b},cs}(q), t+c\alpha_{a,b}s)$$
which is equivalent to the first one, i.e., $Z_{ab}=Z_{ab'}$.

\section{Conclusions}\label{sec:con}
We have presented the general geometric theory of Routh reduction for mechanical systems with one cyclic variable. The problem of many cyclic variables, e.g. systems invariant with respect to group action on the configuration manifold, or, more generally, systems invariant with respect to more than one-dimensional distribution on the configuration manifold (see e.g. \cite{CM,LGA}), we postpone to further publications.  There is one more important line of study that should be pursued from geometric point of view, i.e. systems which are almost invariant with respect to a vector field. Almost invariant meaning that Lagrangian of the system changes by complete derivative when differentiated by the tangent lift of a vector field as discussed e.g. in \cite{LCV}.

%%%%%%%%%%%%%%%%%%%%%%%%%%%%%%%%%%%%%%%%%%%%%%%%%%%%%%%%%%%%%%%%%%%%%%%%%%%%%%%%%%%%%%%%%%%%%%%%%%%%%%%%%%%%%%%%%%%%%%

%%%%%%%%%%%%%%%%%%%%%%%%%%%%%%%%%%%%%%%%%%%%%%%%%%%%%%%%%%%%%%%%%%%%%%%%%%%%%%%%%%%%%%%%%%%%%%%%%%%%%%%%%%%%%%%%%%%%%%%

\end{document}